\documentclass[notitlepage,a4paper]{article}

\usepackage[utf8]{inputenc}
\usepackage[a4paper,margin=1.25in]{geometry}%
\usepackage{lmodern}%

\usepackage{amsmath,amssymb,amsfonts,amsthm}%
\usepackage[english]{babel}%
\usepackage{graphicx}%
\usepackage[numbers,sort]{natbib}%
\usepackage{mathtools}
\usepackage[title]{appendix}%
\usepackage{xcolor}%
\usepackage{manyfoot}%
\usepackage{authblk}%
\usepackage{booktabs}%
\usepackage{algorithm}%
\usepackage{algpseudocode}%
\usepackage{tabularx}%
\usepackage{diagbox}%
\usepackage{multirow}%
\usepackage{microtype}%
\usepackage[colorlinks=true,urlcolor=blue!60!black,citecolor=blue!60!black,linkcolor=blue!60!black,hypertexnames=false]{hyperref}%
\usepackage{cleveref}%

\theoremstyle{plain}%
\newtheorem{theorem}{Theorem}%
\newtheorem{proposition}[theorem]{Proposition}%
\newtheorem{lemma}[theorem]{Lemma}%

\newcommand{\suchthat}{\; : \;}
\DeclareMathOperator{\triangleup}{\vartriangle}
\makeatletter
\providecommand{\bigsqcap}{%
	\mathop{%
		\mathpalette\@updown\bigsqcup
	}%
}
\newcommand*{\@updown}[2]{%
	\rotatebox[origin=c]{180}{$\m@th#1#2$}%
}
\makeatother%

\Crefname{theorem}{Theorem}{Theorems}%
\Crefname{lemma}{Lemma}{Lemmas}%
\Crefname{proposition}{Proposition}{Propositions}%

\crefname{theorem}{theorem}{theorems}%
\crefname{lemma}{lemma}{lemmas}%
\crefname{proposition}{proposition}{propositions}%

\AddToHook{env/lemma/begin}{\crefalias{theorem}{lemma}}%
\AddToHook{env/proposition/begin}{\crefalias{theorem}{proposition}}%

\begin{document}

\title{On the enumeration of Tarski fixed points\thanks{A Java implementation of the enumeration algorithms is available in the netroles library~\cite{m-n-24}.}}

\author{Julian M\"uller}

\affil{Social Networks Lab, ETH Z\"urich\\Institute of Computing, Universit\`{a} della Svizzera italiana\\julian.mueller@gess.ethz.ch}

\date{April 26, 2026}

\maketitle

\begin{abstract}
	We study the problem of enumerating Tarski fixed points on finite lattices. We derive query complexity lower bounds for finding three or more Tarski fixed points of isotone maps and the subclasses of increasing and decreasing isotone maps. Specifically, we show that any deterministic or bounded-error algorithm must perform asymptotically as many queries in the worst case as the lattice width for isotone maps, which is exponential for the lattice of binary relations and other relevant lattices.
	
	We also present two enumeration algorithms for fixed points of increasing or decreasing isotone maps based on depth-first and flashlight search. Both algorithms run in polynomial space on polynomial-height lattices, but are particularly suitable in terms of applicability and runtime performance on different lattices, as they build on differing properties of the underlying lattice.
	
	Finally, we discuss the enumeration of Tarski fixed points on the lattices of binary relations, quasiorders and equivalences, demonstrating that the presented algorithms run in polynomial space on these lattices and perform with polynomial delay whenever the problem of finding three or more fixed points is neither NP-hard nor has an exponential query lower bound. We exemplify how these results can be used to list instances of various models of behavioral or role equivalence, specifically deriving a polynomial-space algorithm that enumerates bisimulations with $\mathcal O(n^3m)$ delay on a transition system with $n$ states and $m$ transitions.
\end{abstract}

\section{Introduction}\label{section:introduction}

Models of behavioral or role equivalence, like the linked concepts of  bisimulation~\citep{m-cc-89,p-cais-81}, structural relatedness~\citep{s-se-78} and regular equivalence~\citep{wr-gshnr-83}, have been used to describe alike states in labeled transition systems or compare the structural similarity of actors in social networks. Instances of such models can commonly be described as fixed points of an efficiently computable isotone map on the underlying lattice~\citep{m-cc-89,mb-er-22,s-obc-09}, and so the existence of minimum or maximum instances as well as the lattice structure of instances follow from the renowned fixed point theorem of Tarski~\citep{t-lftia-55}.

Since the maximal and minimal instances are often trivial or non-informative on structures like undirected graphs~\citep{be-care-89}, social scientists have shown interest in finding other fixed points as potential salient representations of network structure as well~\citep{bdf-oare-92,bbe-ir-89,be-care-89}.

In this article, we study the computational problem of listing such Tarski fixed points on finite lattices in general, but also on the relational lattices of equivalences, quasiorders and binary relations specifically, as these relational lattices commonly underlie models of behavioral or role equivalence (e.g., regular equivalence~\cite{wr-gshnr-83} are equivalences, and simulation preorders~\cite{bks-spsr-12} are quasiorders, and bisimulations~\cite{m-cc-89} are binary relations).

\subparagraph*{Contributions and outline.} We examine the hardness of the Tarski fixed point enumeration problem in \Cref{section:hardness} and show new lower bounds on finding three or more fixed points in terms of query complexity; specifically, we prove that the number of queries is lower-bounded by the width of the lattice for fixed points on isotone maps, and similarly obtain lower bounds for the subclasses of increasing and decreasing isotone maps.

In \Cref{section:algorithm}, we suggest two suitable algorithms for enumerating Tarski fixed points of non-increasing and non-decreasing isotone maps based on depth-first (which generally improves on a similar algorithm proposed by Echenique~\citep{e-faegs-07}) and flashlight search (see, e.g., \cite{ms-eespco-16,rt-bbalcpst-75}). These enumeration algorithms run in polynomial space on polynomial-height lattices, but pose differing structural requirements and have different performance characteristics, making them suitable for different lattices and subclasses of isotone maps. In \Cref{section:relationallattices}, we apply both of them in the appropriate cases to the lattices of equivalences, quasiorders and binary relations, and also transfer the query complexity results to these lattices. Whenever the query complexity bounds and algorithms' runtime bounds are super-polynomial on these lattices, we also prove that finding three or more fixed points is NP-hard. Furthermore, we discuss the application and implications of these results to bisimulation and other models of behavioral equivalence; for example, we obtain a new polynomial-space algorithm for listing bisimulations that runs with a delay of $\mathcal O(n^3m)$ on transition systems with $n$ states and $m$ transitions. Finally, we conclude in \Cref{section:conclusion}.

\subparagraph*{Related research.} Motivated by the observation that Nash equilibria of super-modular games are Tarski fixed points, there has been recent research on the computation of Tarski fixed points on $d$-dimensional grids $[N]^d$ with $N, d \in \mathbb{N}^+$. Notably, various algorithms to find Tarski fixed points on these grids 
have been proposed~\citep{cl-iubft-22,dqy-cctfp-20,fps-faftf-22}, among which the asymptotically fastest requires $\mathcal O(\log^{\lceil(d+1)/2\rceil} N)$~queries to the isotone  map~\citep{cl-iubft-22}. In turn, a $\Omega(\log^2 N)$ lower bound was shown for the two-dimensional grid in \citep{epry-ttsgc-20}, and any established lower bounds must also apply if only isotone maps with unique fixed points are considered~\citep{cly-rtutbm-23}. Moreover, deciding the uniqueness of a fixed point on a $d$-dimensional grid is known to be coNP-hard 
\citep{dqy-cctfp-20,dy-cdutf-10}. 
Finally, the Tarski fixed point problem has also been examined in the case where the partial order of the underlying lattice is given as an oracle \citep{clt-ctfpt-08}.

Note that it has major complexity implications that grid lattices have exponential height in the size of a lattice element: While finding the least and greatest fixed points of polynomial-time computable isotone maps is an NP-hard problem on the $d$-dimensional grid \citep{epry-ttsgc-20}, they can be found in polynomial time simply by naive iterated application of the isotone map on the lattice of binary relations and other lattices of polynomial height.

Echenique~\cite{e-faegs-07} proposed an algorithm to enumerate Nash equilibria of super-modular games, which can be generalized to Tarski fixed points on other lattices. \Cref{section:algorithm} below presents a related but enhanced enumeration algorithm based on depth-first search, which generally provides stronger runtime and memory guarantees: For example, the depth-first search algorithm runs in polynomial space and delay on the lattice of binary relations, while an adaption of Echenique's does not.

In the literature on behavioral and role equivalence, some approaches have been proposed to list instances of regular equivalences, which are Tarski fixed points on the lattice of equivalences. Specifically, a heuristic lattice exploration method~\citep{bbe-ir-89} and a combinatorial optimization approach (generalized blockmodeling \citep{bdf-oare-92,dbf-gb-05}) have been used to find non-trivial instances. However, neither approach provides a guarantee to find non-trivial equivalences, even if some exist. Moreover, the former approach needs to repeatedly solve graph isomorphism problems, while it has been shown for the latter that the underlying $k$-role assignment problem about finding regular equivalences with exactly $k$ equivalence classes is NP-hard \citep{rs-hhidi-01,fp-cccra-05}. Note that this hardness result does not immediately imply that enumeration of regular equivalences is hard as well.

\section{Preliminaries} \label{section:preliminiaries}

In this section, we briefly introduce the concepts and notation used throughout this paper. For further reading on orderings and lattices, we suggest \cite{dp-ilo-02,g-glt-03} to the reader.

\subparagraph*{Enumeration algorithms.} 
Procedures for listing all solutions to a problem are known as enumeration algorithms. Since the number of solutions can be several magnitudes larger than the input size, the runtime of an enumeration algorithm is usually measured in reference to both input and produced output. Johnson et al.~\citep{jyp-gamis-88} suggested to measure an algorithm's \emph{delay}, the time spent between two listed solutions, and proposed that an algorithm runs with \emph{polynomial delay} if the delay is polynomial in the input size. We write $t(f)$ to denote the time required for evaluating $f$.

\subparagraph*{Orders and lattices.} A \emph{quasiorder} is a reflexive and transitive binary relation, and a \emph{partial order} is an antisymmetric quasiorder. A \emph{lattice} $(L, \preceq)$ is a set $L$ with a partial order ${\preceq} \subseteq L \times L$ which is closed under two binary operations, the \emph{meet} $x \operatorname{\sqcap} y$ and the \emph{join} $x \operatorname{\sqcup} y$, that yield the unique greatest lower bound and least upper bound of $x, y \in L$. For example, the set of binary relations on base set $[k] =\{1,\dots,k\}$ for $k \in \mathbb{N}$ together with the refinement ordering $\subseteq$ on sets forms a lattice with set union and intersection as join and meet.

All considered lattices are finite, and we respectively write $\bot$ and $\top$ for the minimum and maximum lattice elements. The amount of space required to represent a lattice element is denoted by $\operatorname{size}(L)$.

A \emph{chain} is a subset $C \subseteq L$ such that $x \preceq y$ or $y \preceq x$ for all $x, y \in C$ and has \emph{length} $|C|-1$. An \emph{antichain} is a subset $A \subseteq L$ such that any elements $x, y \in A$ with $x \neq y$ are incomparable. The \emph{height} $h(L, \preceq)$ of the lattice is the maximum length of any contained chain, while the \emph{width} $w(L, \preceq)$ is the maximum size $|A|$ of any contained antichain $A$.

If $x \preceq y$ for some $x, y \in L$, we say that $x$ \emph{precedes} $y$ and $y$ \emph{succeeds} $x$. For $x \prec y$, we say $x$ \emph{is covered by} $y$, written $x \lessdot y$, if $x \preceq z \preceq y$ implies $x = z$ or $y = z$ for all elements $z \in L$; $x$ is then a \emph{lower cover} of $y$ and $y$ an \emph{upper cover} of $x$. The number of lower and upper covers of an element $x$ in the lattice are given by $\operatorname{lcv}_{(L, \preceq)}(x)$ and $\operatorname{ucv}_{(L, \preceq)}(x)$, respectively; we omit the subscript if the lattice is clear from context.
The maximum number of lower or upper covers of any element in the lattice are written as
\begin{align*}
	\operatorname{lcv}(L,\preceq) &\vcentcolon= \max_{x \in L} \operatorname{lcv}_{(L,\preceq)}(x) &
	\operatorname{ucv}(L,\preceq) &\vcentcolon= \max_{x \in L} \operatorname{ucv}_{(L,\preceq)}(x)
\end{align*}

\subparagraph*{Isotone maps.}
A map $f: L \rightarrow L$ on a lattice $(L, \preceq)$ is \emph{isotone} if $x \preceq y$ implies $f(x) \preceq f(y)$, \emph{decreasing} if $f(x) \preceq x$, and \emph{increasing} if $x \preceq f(x)$ for all $x,y \in L$. We respectively mark decreasing maps with subscript $\triangledown$ and increasing maps with $\triangleup$ to allude to the direction which such maps move along in the lattice.

\subparagraph*{Fixed points.} An element $x \in L$ is a \emph{fixed point} of map $f: L \to L$ if $x = f(x)$. We write $\operatorname{Fix}(f)$ for the set of fixed points. Tarski's theorem then states for finite lattices:

\begin{theorem}[\citep{t-lftia-55}] \label{theorem:tarski}
	For an isotone map $f: L \to L$ on finite lattice $(L, \preceq)$, $(\operatorname{Fix}(f), \preceq)$ is a lattice such that $\bigsqcup \{ x \suchthat x \preceq f(x)\} \in \operatorname{Fix}(f)$ and $\bigsqcap \{ x \suchthat f(x) \preceq x \} \in \operatorname{Fix}(f)$.
\end{theorem}

\subparagraph*{Limits.} 
For a given map $f: L \rightarrow L$, we write $f^{(t)}$ for the $t$-th composition of $f$, i.e., $f^{(t)}(x) \vcentcolon= f(f^{(t-1)}(x))$ and $f^{(0)}(x) \vcentcolon= x$. We also write $f^*(x)$ for the \emph{limit of $x$ under $f$}, i.e., $f^*(x) = \lim_{t \to \infty} f^{(t)}(x)$. This limit is guaranteed to exist for decreasing and increasing maps on finite lattices; to be precise, it is equal to the maximum or minimum fixed point preceding or succeeding $x$ and reached within a polynomial number of queries to $f$ on polynomial-height lattices. The lemma below formalizes this observation.

\pagebreak
\begin{lemma} \label{lemma:limitandfixedpoints}
	\ 
	\begin{enumerate}
		\item
	For a decreasing isotone map $f_{\triangledown}: L \rightarrow L$ on a finite lattice $(L, \preceq)$, $f^*_{\triangledown}(x) = f^{(h(L, \preceq))}_{\triangledown}(x)$ is the greatest fixed point of $f_{\triangledown}$ preceding $x$.
	\item For an increasing isotone map $f_{\triangleup}: L \rightarrow L$ on a finite lattice $(L, \preceq)$, $f^*_{\triangleup}(x) = f^{(h(L, \preceq))}_{\triangleup}(x)$ is the least fixed point of $f_{\triangleup}$ succeeding $x$.
	\end{enumerate}
\end{lemma}
\begin{proof}
	We only prove the claim for decreasing isotone maps, because a symmetric argument applies to increasing ones.
	
	The limit $f^*_{\triangledown}(x)$ of $x \in L$ must be a fixed point of $f$ and must have been reached after $h(L, \preceq)$ iterations, since each iteration must yield predecessors. Finally, any fixed point $y \preceq x$ precedes $f^*_{\triangledown}(x)$: Since $y \preceq x$, isotonicity implies $y = f_{\triangledown}(y) \preceq f_{\triangledown}(x)$, so $y \preceq f_{\triangledown}(x)$. Repeating this argument, we obtain $y \preceq f^{(t)}_{\triangledown}(x)$ for any $t \geq 1$, hence $y \preceq f^{(h(L, {\preceq}))}_{\triangledown}(x) = f^*_{\triangledown}(x)$.
\end{proof}

While limits under decreasing or increasing isotone maps can be found through iterative application of the map, we note that this is not necessarily optimal. For example, much more efficient algorithms are known for some models of role or behavioral equivalence; e.g., see \citep{pt-tpra-87} for regular or bisimulation equivalence, \citep{cc-pg-82,pt-tpra-87} for equitable partitions, \citep{hhk-csfig-95} for simulations and bisimulations, or \citep{bp-tdine-95} for ready simulations.

\section{Query complexity} \label{section:hardness}

In this section, we examine the hardness of listing Tarski fixed points by analyzing the difficulty of deciding whether an isotone map $f: L \rightarrow L$ on a lattice $(L, \preceq)$ has $k$ or more fixed points for some constant $k \in \mathbb{N}$ in terms of query complexity: Isotone maps are given as \emph{oracles}, which means an algorithm can only discern the structure of $f$ by performing queries to evaluate $f(x)$ for elements $x$ or exploiting prior knowledge on the general structure of $f$ (such as its isotonicity). In this section, we show worst-case lower and upper bounds on the number of queries to the isotone map for deterministic or bounded-error randomized algorithms solving the fixed point decision problem within this setting. A bounded-error randomized algorithm is a (Monte Carlo) algorithm that must output the correct result for each input with at least probability $\frac{2}{3}$.\footnote{One might wonder why this question is not analyzed in terms of NP-hardness in this section as well. The main reason is that NP-hardness proofs in this lattice-agnostic setting require additional assumptions on the existence of specific substructures in the underlying lattice and how easy they are to work with computationally in particular ways. These assumptions are very technical, and it is hard to assess their significance and implications. Rather, NP-hardness results and proofs are easier to express and interpret for specific kinds of lattices, as we do in \Cref{section:relationallattices} for the lattices of equivalences, quasiorders and binary relations. Furthermore, the proofs there exemplify how the proof techniques for query complexity in this section can often be adapted to the question of NP-hardness as well.}

\subparagraph*{Up to two fixed points.} By Tarski's theorem (\Cref{theorem:tarski}), we immediately have:

\begin{proposition} \label{proposition:1fixedpoint_oracle}
	There is a deterministic algorithm that makes
	\begin{itemize}
		\item $0$ queries to decide whether there is a fixed point, and
		\item $\mathcal O(h(L, \preceq))$ queries to find a fixed point of an isotone map $f: L \rightarrow L$.
	\end{itemize}
\end{proposition}

\begin{proof}
	By Tarski's theorem, any isotone map on a finite lattice has a fixed point. Moreover, Tarski's theorem states that the  maximum fixed point of $f$ is also the maximum fixed point of $f_\triangledown: x \mapsto x \operatorname{\sqcap} f(x)$. By \Cref{lemma:limitandfixedpoints}, we find this maximum fixed point $f^*_{\triangledown}(\top)$ after at most $h(L, \preceq)$ applications of $f_{\triangledown}$.
\end{proof}

Two fixed points, namely the maximum and minimum, can also be found by performing a similar number of queries.

\begin{proposition} \label{proposition:2fixedpoint_oracle}
	There is a deterministic algorithm that decides the existence of two or more fixed points in $\mathcal O(h(L, \preceq))$ queries.
\end{proposition}
\begin{proof}
	We can find the maximum fixed point with $\mathcal O(h(L, \preceq))$ queries as described in the proof of \Cref{proposition:1fixedpoint_oracle}. By a symmetric procedure, we can also search the minimum fixed point. Finally, we decide on the number of fixed points: If the maximum and minimum fixed points are equal, then there is only one fixed point, otherwise there are at least two.
\end{proof}

\subparagraph*{Three or more fixed points.} The necessary number of queries can increase drastically when testing for three or more fixed points. We first establish the lattice width as a lower bound for the class of arbitrary isotone maps on a lattice $(L,\preceq)$. This result implies an exponential lower bound for many relevant lattices.

\begin{theorem} \label{theorem:3fixedpoints_oracle_isotone}
	Any deterministic or bounded-error randomized algorithm needs $\Omega(w(L, \preceq))$ queries in the worst case to decide whether an isotone map has $k\geq3$ fixed points.
\end{theorem}

\begin{proof}
	Let $A$ be a largest antichain in lattice $L$. We now describe a class of suitable isotone maps with two or three fixed points; a proof for $k>3$ fixed points can be obtained by adding trivial fixed points in the constructed map close to the lattice maximum $\top$ or minimum $\bot$ or by choosing $k-3$ elements of $A$ as additional fixed points. Note that selecting elements of $A$ as additional fixed points does not affect the asymptotics, since $k$ is a constant.
	
	Since $A$ is maximal, any lattice element $x \in L \setminus A$ must either precede or succeed some elements in the antichain. Now consider the set of isotone maps of the following form:
	\begin{itemize}
		\item Lattice maximum $\top$ and minimum $\bot$ are fixed points.
		\item If $x \in L \setminus A$ precedes some $y \in A$, then map it to $\bot$, else map it to $\top$.
		\item At most one element $x \in L$ in the antichain is mapped to itself, all other antichain elements are mapped to $\bot$.
	\end{itemize}
	It is easy to verify that all maps of this form are isotone. Furthermore, these maps are identical everywhere except for the existence and position of a third fixed point among the antichain elements. Thus, deciding the fixed point problem for maps of the constructed form is equivalent to performing an unstructured search on the antichain. Since any deterministic and bounded-error randomized algorithm requires a linear number of queries to perform an unstructured search, this establishes a lower bound of $\Omega(|A|) = \Omega(w(L,\preceq))$.
\end{proof}

By adapting this proof to other suitable antichains, we also obtain lower bounds for increasing and decreasing isotone maps. Note that these translate into lower bounds for deciding the existence of three or more fixed points if we substitute $s$ with $\top$ or $\bot$ in the following theorem.

\begin{theorem} \label{theorem:3fixedpoints_oracle_decreasingisotone}
	Any deterministic or bounded-error randomized algorithm needs $\Omega({\operatorname{lcv}(s)})$ (respectively $\Omega({\operatorname{ucv}(s)})$) queries to decide whether a decreasing (increasing) isotone map has $k\geq 2$ fixed points strictly preceding (succeeding) some lattice element $s \in L$.
\end{theorem}
\begin{proof}
	Since the proofs are symmetrical, we show the bound for decreasing isotone maps only. We construct a set of maps along the lines of the proof of \Cref{theorem:3fixedpoints_oracle_isotone}: We select the set of lower covers of the selected element $s$ as the antichain $A$ instead of the largest antichain, and we additionally make $s$ and all lattice elements not preceding $s$ fixed points. The resulting maps are isotone and decreasing. Since deciding the fixed point problem requires an unstructured search among the lower covers of $s$,  $\Omega({\operatorname{lcv}(s)})$ queries are necessary.
\end{proof}

\section{Enumeration algorithms} \label{section:algorithm}

The results on query complexity in the previous section suggest that searching fixed points of decreasing or increasing maps is easier, as the obtained lower bounds that are often smaller than those of arbitrary isotone maps by orders of magnitudes. We now present two enumeration algorithms that run significantly faster than extensive search by exploiting the additional properties of such maps. For the sake of conciseness, we describe the algorithms for decreasing isotone maps only, however, it is straightforward to rewrite them for increasing maps by replacing order-related concepts with those of opposite meaning (such as lower by upper covers or maxima by minima).

The first algorithm below lists fixed points by performing a combined depth-first search on the graph structures defined by the fixed point lattice and underlying lattice. The other algorithm applies the classic flashlight search~(e.g., \cite{ms-eespco-16,rt-bbalcpst-75}) on lattices whose elements are decomposable into multiple dimensions and projections to fewer dimensions effectively partition the underlying lattice in a beneficial way (as explained below). Since these algorithms build on different structural properties of the underlying lattice, one can be more performant than the other for different kinds of lattices and isotone maps, if both are applicable at all. In \Cref{section:relationallattices} for example, we discuss when each algorithm achieves better runtime bounds on the lattices of binary relations, quasiorders or equivalences.

\subsection{Depth-first search}

A lattice can be interpreted as a graph, with directed edges between elements according to their ordering in the lattice. One approach to listing fixed points is to traverse the graphs of the fixed point lattice and the underlying lattice suitably in alternation. This search strategy operates for a decreasing map $f_{\triangledown}: L \to L$ as follows:
\begin{enumerate}
	\item The maximum fixed point $f^*_{\triangledown}(s)$ preceding a starting element $s \in L$ is marked as discovered.
	\item While there is still a fixed point that is marked as discovered but has not been output yet, we choose one such fixed point $x$ and output it. For each lower cover $y \lessdot x$ in the underlying lattice $(L, \preceq)$, we identify its limit $f^*_{\triangledown}(y)$ (i.e., the maximum fixed point preceding $y$), and mark it discovered if it has not been marked before.
\end{enumerate}

In this way, all fixed points on finite lattices preceding $s$ are enumerated, no matter the order in which fixed points are processed. However, the processing order affects runtime bounds and space requirements. For example, Echenique's enumeration algorithm for Nash equilibria of supermodular games~\citep{e-faegs-07} applies this search strategy and processes fixed points in breadth-first order. However, this processing offers few benefits in memory and runtime requirements compared to any other processing order: Lots of discovered fixed points must still be stored explicitly to avoid repeated processing of fixed points, and the worst-case delay bound stays the same (neglecting the cost of the processing itself).

On the other hand, searching in depth-first order can have major benefits:

\begin{enumerate}
	\item The number of discovered but not completely processed fixed points is bounded by the height of the fixed point lattice, implying much better delay bounds if the lattice height is small.
	\item The set of discovered and processed fixed points follow from the previously searched regions in the underlying lattice, and these in turn can be compactly represented through the active path of the depth-first search. Thus, memory requirements are reduced to storing information about the active search path.
\end{enumerate}

From the active path, it can be deduced as follows whether a fixed point must have already been discovered before: Let $\sqsubset_x$ denote the order in which lower covers of a fixed point $x \in L$ are enumerated, i.e., $y \sqsubset_x z$ means that $y \lessdot x$ is enumerated strictly before $z \lessdot x$ at fixed point $x$. For a fixed point $x$ and the last enumerated lower cover $y \lessdot x$, we can check using $\sqsubset_x$ whether a lattice element $z$ must have been found by depth-first search at this fixed point for any prior processed cover: This is the case if and only if there is a lower cover $w \lessdot x$ that precedes $z$, $z \preceq w$, and that is enumerated before $y$ among lower covers, $w \sqsubseteq_x y$. By performing such tests along the whole active path, we can determine if a fixed point has been found before or not.

\subparagraph*{Algorithm.} The enumeration algorithm is given in \Cref{alg:nincisotoneenumeration2}, where the argument $x$ of $\textsc{depthfirstsearch}()$ denotes the most recently found fixed point, while $P$ represents the active path of the search as pairs of fixed points and their most recently enumerated lower covers. The algorithm recurses the depth-first search on a fixed point immediately below $x$ when it has not been discovered by some other depth-first trajectory before (as deduced from the active path), thus guaranteeing that all fixed points preceding $s$ must have been discovered when the algorithm completes. 

\algrenewcommand\algorithmicrequire{\textbf{Input:}}
\algrenewcommand\algorithmicensure{\textbf{Output:}}
\begin{algorithm}[ht]
	\caption{Depth-first search enumeration algorithm\label{alg:nincisotoneenumeration2}}
	\begin{algorithmic}[1]
		\Require{Decreasing isotone map $f_{\triangledown}: L \rightarrow L$ defined on lattice $(L, \preceq)$ \newline
			Upper bound $s \in L$ on fixed points to enumerate \newline (e.g., $\top$ for enumerating all fixed points)}
		\Ensure{Fixed points of $f_{\triangledown}$ preceding $s$}
		
		\State start enumeration by calling \Call{depthfirstsearch}{$f^*_{\triangledown}(s), \emptyset$}\label{algline:nincisoenumeration2:start}
		
		\Function{depthfirstsearch}{fixed point $x$, active path $P$}
		\State \textbf{output} $x$\label{algline:nincisoenumeration2:output}
		\For{$y \lessdot x$ in order $\sqsubset_x$}\label{algline:nincisoenumeration2:loop}
		\State{$P' \gets P \cup \{(x,y)\}$}
		\State $z \gets f^*_{\triangledown}(y)$
		\If{for all $(u, v) \in P'$ and $w \lessdot u$ with $w \sqsubset_{u} v$ we have $z \not\preceq w$}\label{algline:nincisoenumeration2:test}
		\State \Call{depthfirstsearch}{$z, P'$}\label{algline:nincisoenumeration2:recursivecall}
		\EndIf
		\EndFor
		\EndFunction
	\end{algorithmic}
\end{algorithm}

The algorithm's properties are summarized in \Cref{theorem:nincisotoneenumeration2}. Note that the runtime bounds assume lower covers of a lattice element are listed in $\mathcal O(\operatorname{lcv}(L, \preceq\nolinebreak)\operatorname{size}(L))$ time, and that we can test for lattice elements $x, y, \in L, z \lessdot y$ in $\mathcal O(\operatorname{size}(L))$ (amortized) time whether there is  $w \lessdot y$ such that $x \preceq w \sqsubset_x z$. We also assume that these operations need at most $\mathcal O(\operatorname{size}(L))$ additional space. These bounds can be met on the lattices of equivalences, quasiorders, binary relations, and many other lattices.

\begin{theorem} \label{theorem:nincisotoneenumeration2}
	\Cref{alg:nincisotoneenumeration2} enumerates all fixed points of a decreasing isotone map $f_{\triangledown}$ preceding $s \in L$ with delay of $\mathcal O(h(\operatorname{Fix}(f_{\triangledown}), \preceq)\operatorname{lcv}(L, \preceq)(t(f^*_{\triangledown})+h(\operatorname{Fix}(f_{\triangledown}), {\preceq})\operatorname{size}(L)))$ and in space $\mathcal O(h(\operatorname{Fix}(f_{\triangledown}), {\preceq})\operatorname{size}(L))$, spending on average $\mathcal O(\operatorname{lcv}(L, \preceq)(t(f^*_{\triangledown})+h(\operatorname{Fix}(f_{\triangledown}), {\preceq})\operatorname{size}(L)))$ time per output.
\end{theorem}

\begin{proof}
	\emph{Correctness.} 
	The algorithm searches all fixed points immediately preceding $x$ in the fixed point lattice and then recurses on these immediately preceding fixed points to find all fixed points preceding $x$, avoiding all those fixed points during the search that have been discovered. We show this by proving the following algorithm invariant by induction on the number of predecessors of $x$:
	
	At the start of an iteration of the for loop in line~\ref{algline:nincisoenumeration2:loop}, a fixed point $p$ has been output if and only if there is a pair of fixed point and lower cover $(u,v) \in P \cup \{(x,y)\}$ on the active search path s.t. $u=p$ or there is $w \lessdot u~\text{s.t.}~p \preceq w\sqsubset_{u} v$.
	
	Clearly, this invariant holds for the fixed point $x = \bot$.
	
	Now suppose that the invariant is satisfied for the current iteration of the for loop in line~\ref{algline:nincisoenumeration2:loop}. This implies that fixed point $z = f^*_{\triangledown}(y)$ and its predecessors have been output before if there exist $(u, v) \in P'$ and $w \lessdot u$ such that $z \preceq w \sqsubset_{u} v$. Otherwise, \textsc{depthfirstsearch}() is called in line~\ref{algline:nincisoenumeration2:recursivecall}. Since the invariant is also fulfilled for the first iteration in this recursive call, we have by induction that a fixed point $p$ has been output after this call to \textsc{depthfirstsearch}() if and only if $p \preceq z$ or  there exist $(u, v) \in P'$ and $w \lessdot u$ such that $z \preceq w \sqsubset_{u} v$. Since $z$ is the greatest fixed point preceding $y$ (\Cref{lemma:limitandfixedpoints}), this becomes equivalent to the invariant when moving to the next lower cover of $x$ in order $\sqsubset_{x}$ at the start of the next iteration.
	
	After the for loop has ended, it follows from the invariant that fixed points $p$ have been output if and only if $p \preceq x$ or there exist $(u, v) \in P$ and $w \lessdot u$ such that $z \preceq w \sqsubset_{u} v$. Thus, all fixed points preceding $s$ must have been output by the call to \textsc{depthfirstsearch}() in line~\ref{algline:nincisoenumeration2:start}, as the invariant is satisfied at the start of this call.
	
	\emph{Space and running time.} \textsc{depthfirstsearch}() is called in line~\ref{algline:nincisoenumeration2:recursivecall} for a fixed point preceding $x$, so the call tree has at most depth $h(\operatorname{Fix}(f_{\triangledown}), {\preceq})+1$.
	Hence, \Cref{alg:nincisotoneenumeration2} requires $\mathcal O(h(\operatorname{Fix}(f_{\triangledown}), {\preceq})\operatorname{size}(L))$ space, provided that memory to represent the active path $P$ is shared among recursive calls.
	
	Each call leads to at most $\operatorname{lcv}(L, \preceq)$ queries for limits under $f_{\triangledown}$, and a fixed point is output at the beginning of a call to search(). Moreover, the remaining running time is dominated by the tests in line~\ref{algline:nincisoenumeration2:test} to check if the found fixed point is newly discovered, of which there are at most $h(\operatorname{Fix}(f_{\triangledown}), {\preceq})$. Hence,  $\mathcal O(\operatorname{lcv}(L, \preceq)(t(f^*_{\triangledown})+h(\operatorname{Fix}(f_{\triangledown}), {\preceq})\operatorname{size}(L)))$ time is spent per output on average. Finally, we might have to fully backtrack the active path before terminating or finding a new fixed point, implying a delay of $\mathcal O(h(\operatorname{Fix}(f_{\triangledown}), {\preceq})\operatorname{lcv}(L, {\preceq})(t(f^*_{\triangledown})+h(\operatorname{Fix}(f_{\triangledown}), {\preceq})\operatorname{size}(L)))$.
\end{proof}

\subsection{Flashlight search}

Flashlight search can offer an alternative strategy to enumerate fixed points, if the elements of the lattice are composed of several dimensions and dividing the elements according to their projections to some fixed lower number of dimensions yields a partition into convex subsets. As we show in this section, this can yield a performant enumeration strategy if these convex subsets have few maximal elements that are easy to find.

\subparagraph*{Notation.} We introduce some additional notation to describe the algorithm below. For the sake of notational simplicity, we assume in this section that lattice elements are represented as $p$-dimensional vectors, i.e., $L \subseteq A^p$ for some base set $A$. For a set of indices $I = \{i_1, \dots, i_k\} \subseteq [p]$ with $i_1 < \dots < i_k$ and a vector $x$, $x_I$ denotes the vector $(x_{i_1}, \dots, x_{i_k})$ and is called the \emph{projection} of $x$ on $I$. For a vector $x$ of dimension $q \leq p$, we denote all lattice elements extending it by $\operatorname{ext}(x) = \{ y \in L \suchthat y_{[q]} = x \}$, and write $\operatorname{maxext}(x)$ for the set of maximal elements in $\operatorname{ext}(x)$ according to the lattice ordering. Finally, $\operatorname{widen}(x) = \{ (x, a) \suchthat a \in A \wedge \operatorname{ext}(x,a) \neq \emptyset \}$ denotes the set of widenings of $x$, which are those vectors of $q+1$ dimensions that are projections of extensions of $x$. We write $|{\operatorname{widen}}|$ and $|{\operatorname{maxext}}|$ for the maximum number of widenings and maximal extensions of any lower-dimensional projection of a lattice element $x\in L$.

\subparagraph*{Algorithm.} Fixed points of a decreasing isotone map $f_{\triangledown}$ are found by slowly building up lattice elements, adding one dimension at a time to lower-dimensional projections. For each projection, the algorithm solves the extension problem: Does there exist a fixed point of $f_{\triangledown}$ that has this projection? If not, it skips generating higher-dimensional extensions of this projection, avoiding unnecessary work.

Generally, the viability of flashlight search heavily depends on the difficulty of this extension problem. For this specific task, we can rewrite it into a more manageable problem if the extensions $\operatorname{ext}(v)$ form a convex set for any projection $v$ of dimension $q \leq p$, i.e., for any extensions $x, y \in \operatorname{ext}(v)$ and lattice element $z \in L$, we have that $x \preceq z \preceq y$ implies $z_{[q]} = v$. Under this assumption, the following lemma holds:

\begin{lemma} \label{lemma:flashlight_maximalelements}
	Let $f_{\triangledown}: L \rightarrow L$ be a decreasing isotone map. For a vector $v$ of dimension $q$, there is a fixed point of  $f_{\triangledown}$ in  $\operatorname{ext}(v)$ if and only if there is a maximal element $y$ of $\operatorname{ext}(v)$ s{.}t{.} $v = (f^*_{\triangledown}(y))_{[q]}$.
\end{lemma}
\begin{proof}
	Suppose there is a fixed point $x \in \operatorname{ext}(v)$, and let $y$ be a maximal element of $\operatorname{ext}(v)$ s{.}t{.} $x \preceq y$. By induction on isotonicity and the decreasing property of $f_{\triangledown}$, we have $x = f^*_{\triangledown}(x) \preceq f^*_{\triangledown}(y) \preceq y$. Since $\operatorname{ext}(v)$ is convex, it follows that $(f^*_{\triangledown}(y))_{[q]} = v$.
\end{proof}

\Cref{alg:backtracksearch} describes the procedure in more detail, and \Cref{theorem:backtracksearch} summarizes its properties under the assumption that generating the next projection with one more dimension as well as enumerating all maximal extensions takes at most $\mathcal O(|{\operatorname{maxext}}|\cdot t(f^*_{\triangledown}))$ time.

\begin{algorithm}[ht]
	\caption{Flashlight search enumeration\label{alg:backtracksearch}}
	\begin{algorithmic}[1]
		\Require{Decreasing isotone map $f_{\triangledown}: L \rightarrow L$ defined on lattice $(L, \preceq)$ with $L \subseteq A^p$}
		\Ensure{Fixed points of $f_{\triangledown}$}
		
		\State start enumeration by calling \Call{flashlightsearch}{$()$}\label{algline:backtracksearch:start}
		
		\Function{flashlightsearch}{projection $x \in A^q$ with $0 \leq q < p$}
		\For{$y \in \operatorname{widen}(x)$}\label{algline:backtracksearch:outerloop}
		\If{there is $z \in \operatorname{maxext}(y)$ s{.}t{.} $y = (f_{\triangledown}^*(z))_{[q+1]}$}\label{algline:backtracksearch:outerif}
		\If{$q + 1 = p$}
		\State \textbf{output} $y$
		\Else
		\State \Call{flashlightsearch}{$y$}
		\EndIf
		\EndIf
		\EndFor
		\EndFunction
	\end{algorithmic}
\end{algorithm}

\begin{theorem} \label{theorem:backtracksearch}
	\Cref{alg:backtracksearch} enumerates all fixed points of $f_{\triangledown}$ with $\mathcal O(p \nolinebreak \cdot \nolinebreak |{\operatorname{widen}}|\cdot |{\operatorname{maxext}}|\cdot t(f^*_{\triangledown}))$ delay.
\end{theorem}
\begin{proof}
	The algorithm clearly outputs each fixed point once. In the worst case, the algorithm performs the following steps between any two outputs: The algorithm evaluates projections with successively smaller number of dimensions following by projections with successively larger dimension to find the next fixed point. This means the algorithm potentially enumerates the widenings of at most two vectors for each dimension $1 \leq q \leq p$. For each widening in turn, it might enumerate all maximal extensions and evaluate $f^*_{\triangledown}$ for each maximal extension. This gives the delay bound.
\end{proof}

\subsection{Handling arbitrary isotone maps}

\Cref{alg:nincisotoneenumeration2,alg:backtracksearch} utilize that supplied maps are increasing or decreasing, and thus are not immediately applicable for isotone maps which are neither decreasing nor increasing. However, an isotone map $f: L \rightarrow L$ is naturally associated with the decreasing isotone map $f_{\triangledown}: x \mapsto f(x) \operatorname{\sqcap} x$ and the increasing isotone map $f_{\triangleup}: x \mapsto f(x) \operatorname{\sqcup} x$. Since $\operatorname{Fix}(f) = \operatorname{Fix}(f_{\triangledown}) \cap \operatorname{Fix}(f_{\triangleup})$, we can find all fixed points by applying an enumeration algorithm to one of these associated maps and selecting only fixed points of $f$ in the output. This filtering approach is  used by Echenique's Nash equilibria enumeration algorithm~\cite{e-faegs-07}.

However, many fixed points of the associated increasing or decreasing maps might be found between actual fixed points of $f$. Hence, the runtime bounds of the enumeration algorithms do not translate into similar bounds for fixed point enumeration of general isotone maps.

\section{Fixed points on relational lattices} \label{section:relationallattices}

In this section, we consider the issue of the enumeration of Tarski fixed points on the relational lattices of equivalences, quasiorders and binary relations, which notions of behavioral or role equivalence are commonly defined on. Specifically, we transfer the algorithms and query complexity results in \Cref{section:hardness,section:algorithm} to these lattices, but also add results on the issue of NP-hardness. Finally, we discuss how these findings apply to models of behavioral equivalence.

\subparagraph*{Properties of quasiorders, equivalences and relational lattices.}  For simplicity, we define $A=[n]$ as the base set of the considered relations for some $n\in\mathbb{N}_+$. An \emph{equivalence} is a symmetric quasiorder. For an equivalence relation ${\approx}\subseteq A\times A$ on a base set $A$, the \emph{equivalence class} of $x\in A$ is the set $[x]_\approx\vcentcolon=\{y\in A\,:\; y\approx x\}$. An \emph{indifference class} of a quasiorder $\preceq$ is an equivalence class of equivalence ${\preceq} \cap {\preceq}^{-1}$, where ${\preceq}^{-1}$ is the inverted quasiorder $x \preceq^{-1} y \iff y \preceq x$.

We write $\mathcal P$ for the set of equivalences, $\mathcal Q$ for the set of quasiorders, and $\mathcal R$ for the set of binary relations defined on $A$. These sets form lattices under the refinement ordering $\subseteq$. The meets are given by the set intersection of the elements and the joins by the set union for binary relations and the transitive closure of the set union for equivalences and binary relations.

We observe the following bounds on height, lower and upper covers for these lattices:
\begin{lemma} \label{lemma:predecessorandsuccessorcounts}
	\begin{enumerate}
		\item The lattice of equivalence has height $h(\mathcal P, \subseteq) = n-1$ and there are at most $\operatorname{lcv}(\mathcal P, \subseteq), \operatorname{lcv}(\top) \in \Theta(2^n)$ lower and $\operatorname{ucv}(\mathcal P, \subseteq), \operatorname{ucv}(\bot) \in \Theta(n^2)$ upper covers.
		\item An equivalence with $k$ equivalence classes has at most $2^{n-k}$ lower covers.
		\item The lattice of quasiorders has height $h(\mathcal Q, {\subseteq}) \in \Theta(n^2)$ and there are at most \linebreak $\operatorname{lcv}(\mathcal Q, {\subseteq}), \operatorname{lcv}(\top) \in \Theta(2^n)$ lower and $\operatorname{ucv}(\mathcal Q, \subseteq), \operatorname{ucv}(\bot) \in \Theta(n^2)$ upper covers.
		\item A quasiorder with $k$ indifference classes has at most $2^{n-k+1} + \mathcal O(n^2)$ lower covers.
		\item The lattice of binary relations has height  $h(\mathcal R, {\subseteq}) = n^2-1$ and there are at most \linebreak $\operatorname{lcv}(\mathcal R, {\subseteq}),  \operatorname{lcv}(\top) \in \Theta(n^2)$ lower and $\operatorname{ucv}(\mathcal R, {\subseteq}), \operatorname{ucv}(\bot) \in \Theta(n^2)$ upper covers.
	\end{enumerate}
\end{lemma}

\begin{proof}
	We briefly prove the claims for the most complex case of quasiorders only. Similar arguments also apply to equivalences and binary relations.
	
		The maximum quasiorder $\top$ has $2^n-2$ lower covers, and the minimum quasiorder $\bot$ has $n(n-1)$ upper covers. Furthermore,  consider any quasiorder with $k$ indifference classes of sizes $s_1, \dots, s_k$. Any upper cover must order one pair of indifference classes in some additional way, giving an upper bound of $k(k-1) \leq n(n-1)$ for upper covers. Similarly, any lower cover must either remove the strict ordering between a pair of indifference classes, or split an indifference class into two strictly ordered classes. We thus have at most $k(k-1) + \sum_{i=1}^k (2^{s_i}-2) \leq k(k-1)+2\prod_{i=1}^k 2^{s_i-1} \leq k(k-1) + 2^{n-k+1}$ lower covers.
		
		A chain must consist of quasiorders with mutually distinct numbers of contained pairs, so it can have at most size $n^2$. Furthermore, consider the following way to construct a chain: Starting from the quasiorder with only incomparable elements, we first successively make one element less than all others, then a second element less than all others except the first one, and so on. This yields a chain consisting of $n(n-1)/2$ quasiorders.
\end{proof}

\subsection{Enumeration}

We now apply the algorithms from \Cref{section:algorithm} specifically to list fixed points of isotone maps defined on relational lattices.

\subparagraph*{Encoding the lattice elements.} We use the following natural representations of the lattice elements, which are asymptotically optimal in space and suitable for the enumeration algorithms described in \Cref{section:algorithm}.

\begin{itemize}
	\item We describe binary relations and quasiorders by $n^2$-dimensional binary vectors $v \in \{0, 1\}^{n^2}$ such that $(i,j) \in R$ if and only if $v_{(i-1)n + j} = 1$.
	\item We encode equivalences by $n$-dimensional vectors $v \in \{1, \dots, n\}^n$ with $v_1 = 1$ and $v_k \leq \max_{1 \leq \ell < k} v_\ell + 1$ for $k \in \{2, \dots, n\}$ and interpret such a vector to mean that $i$ and $j$ are equivalent if and only if $v_i = v_j$.
\end{itemize}

With these representations, binary relations and quasiorders have size $\Theta(n^2)$ and equivalences have size $\Theta(n)$.

\subparagraph*{Depth-first search.} This algorithm requires a way to enumerate the lower and upper covers of an  element such that it can be tested efficiently if a lattice element is a descendant of a cover of this lattice element that was enumerated before. This is straightforward to achieve on the relational lattices; see \cite{m-n-24} for an efficient implementation.

\subparagraph*{Flashlight search.} Projections of the lattice element representations above to fewer dimensions are straightforward and it is easy to prove that projections of the same dimension correspond to a partitioning of the original lattice into convex subsets. In addition, it is not difficult to grow projection vectors in these representations as required by flashlight search, although the procedure is slightly complicated for quasiorders by the fact that additional filtering is required to exclude those binary projection vectors that cannot be extended to any quasiorder, i.e., those projections failing transitivity or reflexivity.

The asymptotic running time of \Cref{alg:backtracksearch} heavily depends on the number of minimum or maximum extensions of a projection. It can be shown straightforwardly that there is always only one minimal extension for all three lattices and always only one maximal one on the lattice of binary relations, since the convex subsets obtained by partitioning according to projections are closed under meet for all three relational lattices and closed under join for binary relations:

\begin{itemize}
	\item \emph{Binary relations:} The maximum and minimum extension of a projection are obtained by extending it only with ones (i.e., adding all pairs not fixed by the projection) or only zeros (i.e., adding no such pairs), respectively.
	\item \emph{Quasiorders:} The minimum extension of a projection is given by the reflexive-transitive closure on the relation defined by only those pairs that must be present according to the projection vector.
	\item \emph{Equivalences:} The minimum extension of a projection is obtained by making all remaining elements non-equivalent to all others.
\end{itemize}

However, the following example shows that there can be exponentially many maximal extensions on the lattices of quasiorders and equivalences: Suppose a projection~$v$ on the first $n/2$ elements makes all of these elements non-equivalent. Then we could add any remaining element to one of the $n/2$ equivalence classes of the first half of elements. Each combination of choices results in a different equivalence that is not coarsened by any other projecting to $v$. Thus, there are at least $\Omega((n/2)^{n/2})$ maximal extensions for projections on the lattice of equivalence, and by analogy on the lattice of quasiorders.

\begin{table}[t]
	\caption[Worst-case bounds on space and delay of \Cref{alg:nincisotoneenumeration2}]{Worst-case delay and space bounds achieved for enumerating fixed points of decreasing or non-increasing isotone maps on the relational lattices using \Cref{alg:nincisotoneenumeration2,alg:backtracksearch}. The given bounds assume that $t(f^*) \in \Omega(\operatorname{size}(L))$ for each lattice $L$. \label{table:latticeenumeration:algorithmruntime}}
	\centering
	\begin{tabularx}{\linewidth}{l|c|c|>{\centering\arraybackslash}X|c|c|>{\centering\arraybackslash}X|}
		\multirow{2}{*}{lattice} & \multicolumn{3}{c|}{decreasing isotone map $f^*_{\triangledown}$} & \multicolumn{3}{c|}{increasing isotone map $f^*_{\triangleup}$} \\
		& Alg. & delay & space &  Alg. & delay & space \\\hline
		equivalences & Alg.~\ref{alg:nincisotoneenumeration2} & $\mathcal O(2^n(t(f^*_\triangledown) + n^2))$ & $\mathcal O(n^2)$ & Alg.~\ref{alg:backtracksearch} & $\mathcal O(n^2 \cdot t(f^*_{\triangleup}))$ & $\mathcal O(n)$ \\
		quasiorders & Alg.~\ref{alg:nincisotoneenumeration2} & $\mathcal O(2^n(t(f^*_\triangledown) + n^4))$ & $\mathcal O(n^4)$ & Alg.~\ref{alg:backtracksearch} & $\mathcal O(n^2 \cdot t(f^*_{\triangleup}))$ & $\mathcal O(n^2)$ \\
		binary relations & Alg.~\ref{alg:backtracksearch} & $\mathcal O(n^2 \cdot t(f^*_{\triangledown}))$ & $\mathcal O(n^2)$ & Alg.~\ref{alg:backtracksearch} & $\mathcal O(n^2 \cdot t(f^*_{\triangleup}))$ & $\mathcal O(n^2)$ \\
	\end{tabularx}
\end{table}

\subparagraph*{Summary.} \Cref{table:latticeenumeration:algorithmruntime} summarizes the best worst-case delay and space bounds achieved for the fixed point enumeration on the relational lattices by either of these algorithms in accordance with \Cref{theorem:nincisotoneenumeration2,theorem:backtracksearch}, and also states which of these algorithms achieves each bound. The bounds for decreasing isotone maps on the lattice of equivalences and quasiorders, slightly improving on the general bounds in \Cref{theorem:nincisotoneenumeration2},  follow from the following observations:

\begin{proposition}
	\begin{enumerate}
		\item \Cref{alg:nincisotoneenumeration2} enumerates at most $2^n$ lower covers between discovered fixed points on the lattice of equivalences.
		\item A variant of \Cref{alg:nincisotoneenumeration2} enumerates at most $2^n$ lower covers between discovered fixed points on the lattice of equivalences.
	\end{enumerate}
\end{proposition}
\begin{proof}
	\begin{enumerate}
		\item Assume that the algorithm is searching a fixed point on some active path. It will bracktrack along this path until it discovers the next fixed point or terminates. Following along the active path, the number of equivalence classes increases by at least one from one fixed point to the next. Thus, the path can have at most length $n$, and since an equivalence with $k$ equivalence classes has at most $2^{n-k}$ lower covers (\Cref{lemma:predecessorandsuccessorcounts}), the algorithm searches at most $\sum_{k=1}^n 2^{n-k} < 2^n$ lower covers until it discovers the next fixed point or terminates.
	\item Consider a variant of the algorithm that enumerates the lower covers of a fixed point that split indifference classes before those lower covers that remove an ordering between indifference classes. Due to this order, a class-splitting lower cover can be skipped in recursive calls if even a single ordering between indifference classes was removed along the active path: Any fixed point preceding the class-splitting lower cover would have been discovered before along another active path splitting the classes first, and thus the tests in \Cref{alg:nincisotoneenumeration2}'s line~\ref{algline:nincisoenumeration2:test} does not succeed for any of them. This means that class splits can only occur for the $k=1$-st up to the $n$-th fixed point on the active path, the number of classes always increasing by at least one. The $k$-th such fixed point has at most $2^{n-k+1}+\mathcal O(n^2)$ lower covers (\Cref{lemma:predecessorandsuccessorcounts}). Following along the active path for the remaining fixed points, only orderings between classes are removed from one fixed point to the next. Thus, there can be at most $\mathcal O(n^2)$ additional fixed points on the active path and since class splits cannot discover new fixed points, the algorithm only has to search the $\mathcal O(n^2)$ ordering-removing lower covers of each of these fixed points. Hence, the algorithm examines at most $\sum_{k=1}^{n} (2^{n-k+1} + \mathcal O(n^2)) + \mathcal O(n^2\cdot n^2) \subseteq \mathcal O(2^n)$ lower covers until it finds the next fixed point or terminates. \qedhere
	\end{enumerate}
\end{proof}

Comparison of the worst-case bounds shows that flashlight search yields the best bounds in those four cases in which each projection of lattice elements is only associated with a single maximal or minimal extension, scaling noticeably better than depth-first search. On the other hand, flashlight search has to enumerate $\Omega ((n/2)^{n/2})$ or more maximal extensions for a projection vector in the other two cases, implying exponentially worse bounds than the ones of depth-first search as listed in \Cref{table:latticeenumeration:algorithmruntime}. This showcases that there are situations where either algorithm can have a substantial advantage over the other when they are both applicable to the lattices.

\subsection{Hardness}

We now examine the hardness of finding fixed points on the relational lattices.

\subparagraph*{Query complexity.} On the relational lattices, the results obtained in \Cref{section:hardness} translate to the bounds listed in \Cref{table:latticeenumeration:lowerbounds_oracle}. By comparing these with the worst-case runtime bounds in \Cref{table:latticeenumeration:algorithmruntime} that are achieved by \Cref{alg:nincisotoneenumeration2,alg:backtracksearch}, we see that lower and upper bounds for increasing and decreasing isotone maps are close, only the exponential bounds showing a small discrepancy at first glance. However, there is still a subtle gap between the lower and upper bounds even if they appear to be the same: While the lower bounds count queries to $f$, the upper bounds are expressed in terms of querying limits $f^*$ of map $f$, computation of which might require another $h(L, {\preceq})$ queries to $f$ in turn. This means there remains a gap of a factor up to $n$ for equivalences and $n^2$ for quasiorders and binary relations.

\begin{table}[t]
	\caption[Query lower bounds for deciding the existence of three fixed points]{Worst-case query lower bounds for deciding the existence of three or more fixed points when isotone maps are given as oracles.\label{table:latticeenumeration:lowerbounds_oracle}}
	\centering
	{\renewcommand\tabularxcolumn[1]{m{#1}}
		\begin{tabularx}{\linewidth}{l|c|>{\centering\arraybackslash}X|c}
			\diagbox{lattice\,\,\,\,\,\,\,\,\,\,\,}{maps\,\,\,\,\,\,\,\,\,\,\,} & decreasing isotone & isotone & increasing isotone \\\hline
			\multirow{2}{*}{equivalences} &  $\Omega(2^n)$ &  $\Omega((n/(\log(n)\cdot e))^n)$ & $\Omega(n^2)$ \\
			& (Thm.~\ref{theorem:3fixedpoints_oracle_decreasingisotone}, \Cref{lemma:predecessorandsuccessorcounts}) & (Thm.~\ref{theorem:3fixedpoints_oracle_isotone}, \citep{rd-snsk-69}) & (Thm.~\ref{theorem:3fixedpoints_oracle_decreasingisotone}, \Cref{lemma:predecessorandsuccessorcounts}) \\\hline
			\multirow{2}{*}{quasiorders} &  $\Omega(2^n)$ &  $\Omega((n/(\log(n)\cdot e))^n)$ & $\Omega(n^2)$ \\
			& (Thm.~\ref{theorem:3fixedpoints_oracle_decreasingisotone}, \Cref{lemma:predecessorandsuccessorcounts}) & (Thm.~\ref{theorem:3fixedpoints_oracle_isotone}, \citep{rd-snsk-69}) & (Thm.~\ref{theorem:3fixedpoints_oracle_decreasingisotone}, \Cref{lemma:predecessorandsuccessorcounts}) \\\hline
			\multirow{2}{*}{binary relations} &  $\Omega(n^2)$ &  $\Omega(2^{n^2}/n)$ & $\Omega(n^2)$ \\
			& (Thm.~\ref{theorem:3fixedpoints_oracle_decreasingisotone}, \Cref{lemma:predecessorandsuccessorcounts}) & (Thm.~\ref{theorem:3fixedpoints_oracle_isotone}, \citep{s-sueuem-28}) & (Thm.~\ref{theorem:3fixedpoints_oracle_decreasingisotone}, \Cref{lemma:predecessorandsuccessorcounts}) \\\hline
		\end{tabularx}
	}
\end{table}

\subparagraph*{NP-hardness.} To investigate the NP-hardness of the fixed point problem, we consider it in the \emph{polynomial map setting}, in which the analysis is restricted to succinctly represented maps. A map $f$ is \emph{succinctly represented} if $f$'s memory representation has size polynomial in $n$ and $f(x)$ is computable in polynomial time for each input~$x$. Note that a solving algorithm is not constrained to only querying function values, but can potentially choose to access and analyze the representation of $f$ as well.

Deciding the existence of two fixed points on the relational lattices can be done in polynomial time using the algorithm described in the proof of \Cref{proposition:2fixedpoint_oracle}. However, we show below that the problem is sometimes NP-hard for three or more fixed points, namely in those cases in which \Cref{table:latticeenumeration:algorithmruntime,table:latticeenumeration:lowerbounds_oracle} list exponential bounds in terms of runtime of the enumeration algorithms and query complexity. These NP-hardness results proven below are summarized in \Cref{table:latticeenumeration:hardness}. Observe that the NP-hardness results imply as well that the extension problem underlying flashlight search cannot be solved in polynomial time for decreasing isotone maps on the lattices of quasiorders and equivalences unless $\mathrm{P} = \mathrm{NP}$.

\begin{table}[t]
	\caption[Complexity of deciding the existence of $k$ fixed points for $k \geq 3$.]{Complexity of deciding the existence of $k$ fixed points on relational lattices in the polynomial map setting for $k \geq 3$.\label{table:latticeenumeration:hardness}}
	\centering
	{\renewcommand\tabularxcolumn[1]{m{#1}}
		\begin{tabularx}{\linewidth}{l|c|>{\centering\arraybackslash}X|c}
			{\diagbox{lattice\,\,\,\,\,\,\,\,\,\,\,}{maps\,\,\,\,\,\,\,\,\,\,\,}} & decreasing isotone & isotone & increasing isotone \\\hline
			\multirow{2}{*}{equivalences} & NP-complete & NP-complete & P \\
			& (Thm.~\ref{theorem:3fixedpoints_equivalence}) & (Thm.~\ref{theorem:3fixedpoints_equivalence}) & (\Cref{table:latticeenumeration:algorithmruntime}) \\\hline
			\multirow{2}{*}{quasiorders} & NP-complete & NP-complete & P \\
			& (Thm.~\ref{theorem:3fixedpoints_quasiorder}) & (Thm.~\ref{theorem:3fixedpoints_quasiorder}) & (\Cref{table:latticeenumeration:algorithmruntime}) \\\hline
			\multirow{2}{*}{binary relations} & P & NP-complete & P \\
			& (\Cref{table:latticeenumeration:algorithmruntime}) & (Thm.~\ref{theorem:3fixedpoints_binaryrelation}) & (\Cref{table:latticeenumeration:algorithmruntime}) \\\hline
		\end{tabularx}
	}
\end{table}

The following hardness proofs adapt the proof strategy from \Cref{section:hardness}: They select a suitable antichain in the underlying lattice and construct an isotone map that places additional fixed points on this antichain if and only if an instance of the 3SAT problem is satisfiable.

\begin{theorem} \label{theorem:3fixedpoints_binaryrelation}
	Deciding the existence of $k$ fixed points of an (arbitrary) isotone map is NP-complete on the lattice of binary relations for $k \geq 3$.
\end{theorem}

\begin{proof}
	We reduce \textsc{3SAT} to the three-fixed point problem on the lattice of binary relations; a reduction for more than three fixed points can be obtained by artificially introducing $k-3$ additional trivial fixed points in the constructed map close to the lattice minimum $\bot$.
	
	Suppose we are given a formula $B$ with $m\geq1$ variables. We choose a base set $A$ of size $n$ such that $n^2 \geq 2m$. Let $i_1, \dots, i_{n^2}$ be all possible pairs of elements on $A$. We define a polynomial-time computable map $f: \mathcal R \to \mathcal R$ by mapping a binary relation $R \in \mathcal R$ as follows:
	
	\begin{enumerate}
		\item If $\{i_{2m+1}, \dots, i_{n^2}\} \cap R \neq \emptyset$, then map $R$ to $\top$.
		\item Otherwise, if $i_{\ell} \in R$ and $i_{\ell+m} \in R$ for some $\ell \in \{1, \dots, m\}$, then map $R$ to $\top$.
		\item Otherwise, if $i_{\ell} \notin R$ and $i_{\ell+m} \notin R$ for some $\ell \in \{1, \dots, m\}$, then map $R$ to $\bot$.
		\item Otherwise, if $I=(i_1 \in R, \dots, i_m \in R)$ is a satisfying assignment of formula $B$, then map $R$ to $R$, else map $R$ to $\bot$.
	\end{enumerate}
	It is easy to verify that $f$ is isotone, can be succinctly represented (since $B$ has at most $n^3$ distinct clauses) and has the trivial fixed points $\bot$ and $\top$. We now show that $f$ has at least three fixed points if and only if formula $B$ is satisfiable.

	Let $I=(x_1, \dots, x_m) \in \{0,1\}^m$ be a satisfying assignment of $B$. Then, the relation \[R = \{ i_\ell \suchthat x_\ell = 1, \ell \in \{1, \dots, m\} \} \cup \{ i_{m+\ell} \suchthat x_\ell = 0, \ell \in \{1, \dots, m\} \}\] is a fixed point of $f$. This fixed point is necessarily different from the trivial minimum and maximum fixed points, as $R$ contains exactly $m$ pairs, while the trivial ones contain $0$ or $n^2$.
	
	Conversely, assume that $f$ has three fixed points. Then let $R$ be a fixed point different from the trivial fixed points $\top$ and $\bot$. By construction of $f$, $R$ must fall into the first case of step 4 above. Thus, $I=(i_1 \in R, \dots, i_m \in R)$ is a satisfying assignment of formula $B$.
\end{proof}

\begin{theorem} \label{theorem:3fixedpoints_equivalence}
	Deciding the existence of $k$ fixed points of a decreasing isotone map is NP-complete on the lattice of equivalences for $k \geq 3$.
\end{theorem}

\begin{proof}
	We reduce from \textsc{3SAT} similar to the proof of \Cref{theorem:3fixedpoints_binaryrelation} and briefly describe the construction for $k=3$ only. For a given formula $B$ with $m\geq1$ variables, we choose base set $A=\{v_1, \dots, v_{m+2}\}$ of size $n = m+2$. The reduction is based on a map $f_\triangledown: \mathcal P \rightarrow \mathcal P$, which maps equivalence ${\approx} \in \mathcal P$ as follows:
	\begin{enumerate}
		\item If ${\approx} = \top$, then map $\approx$ to itself.
		\item If $\approx$ has three or more equivalence classes, then map it to $\bot$.
		\item Otherwise, if $I=(v_1 \in [v_n]_{\approx}, \dots, v_m \in [v_n]_{\approx})$ is a satisfying assignment of $B$, then map $\approx$ to itself, else map it to $\bot$.
	\end{enumerate}
	$f_\triangledown$ is decreasing isotone and can be succinctly represented. Finally, it can be shown that $f_\triangledown$ has three or more fixed points if and only if formula $B$ is satisfiable.
\end{proof}

\begin{theorem} \label{theorem:3fixedpoints_quasiorder}
	Deciding the existence of $k$ fixed points of a decreasing isotone map is NP-complete on the lattice of quasiorders for $k \geq 3$.
\end{theorem}

\begin{proof}
	We perform a reduction from 3SAT on the same lines as in the proof of \Cref{theorem:3fixedpoints_binaryrelation} and only briefly describe the construction for $k = 3$ here. For a given formula $B$ with $m \geq 1$ variables, we choose a base set $A = \{v_1, \dots, v_{m+2}\}$ and size $n = m+2$. We construct a map $f_\triangledown: \mathcal Q \rightarrow \mathcal Q$ by mapping a quasiorder ${\lessapprox} \in \mathcal Q$ as follows:
	
	\begin{enumerate}
		\item If ${\lessapprox} = \top$, map it to itself.
		\item Otherwise, if $\lessapprox$ does not consist of two ordered indifference classes, map $\lessapprox$ to $\bot$.
		\item Otherwise, let $X$ be the only maximal indifference class in $\lessapprox$.
		If $I=(v_1 \in X, \dots, v_m \in X)$ is a satisfying assignment of $B$, then map $\lessapprox$ to itself, else map it to $\bot$.
	\end{enumerate}
	This map $f_\triangledown$ is decreasing isotone, can be succinctly represented and has three or more fixed points if and only if $B$ is satisfiable. 
\end{proof}

\subsection{Application: Listing bisimulations and behavioral equivalences}

The algorithms in \Cref{section:algorithm} can be employed to list instances of some models of behavioral and role equivalence. In this section, we mainly focus on the case of listing bisimulations.

A labeled transition system is a triple $(S, \Lambda, \rightarrow)$ of states $S$, labels $\Lambda$ and labeled transitions~${\rightarrow} \subseteq \nolinebreak S \times \Lambda \times S$; we write $p \overset{\alpha}{\rightarrow} q$ to mean $(p, \alpha, q) \in {\rightarrow}$. A binary relation $R \subseteq S \times S$ is called a \emph{bisimulation}~\cite[p.~88]{m-cc-89} if $(p, q) \in R$ implies that
\begin{itemize}
	\item whenever $p \overset{\alpha}{\rightarrow} p'$, there is $q \overset{\alpha}{\rightarrow} q'$ such that $(p',q') \in R$, and
	\item whenever $q \overset{\alpha}{\rightarrow} q'$, there is $p \overset{\alpha}{\rightarrow} p'$ such that $(p',q') \in R$.
\end{itemize}

Equivalently, bisimulations can also be described as the fixed points of the decreasing isotone map $f: \mathcal{R} \rightarrow \mathcal{R}$ given by
\begin{equation} \label{equation:isotonemap_bisimulation}
	(p, q) \in f(R) \iff \left\{\begin{array}{l} (p, q) \in R,~\text{and} \\ \text{for each}~p \overset{\alpha}{\rightarrow} p'~\text{there exists}~p \overset{\alpha}{\rightarrow} q'~\text{s.t.}~(p',q') \in R,~\text{and} \\ \text{for each}~q \overset{\alpha}{\rightarrow} q'~\text{there exists}~p \overset{\alpha}{\rightarrow} p'~\text{s.t.}~(p',q') \in R \end{array}\right.
\end{equation}

By applying \Cref{alg:backtracksearch} to the isotone map above, we can thus enumerate bisimulations with polynomial delay and space, as the following theorem shows.

\begin{theorem}
	Let $(S, \Lambda, \rightarrow)$ be a finite labeled transition system with $n=|S|$ states and $m=|{\rightarrow}|$ transitions. Then \Cref{alg:backtracksearch} can enumerate bisimulations with $\mathcal O(n^3m)$ delay and $\mathcal O(n^2)$ space.
\end{theorem}
\begin{proof}
	Using the algorithm in~\citep{hhk-csfig-95}, the limit of a binary relation $R$ under the decreasing isotone map in \Cref{equation:isotonemap_bisimulation} can be computed in $\mathcal O(nm)$ time and $\mathcal O(n^2)$ space. The bounds on delay and space follow from  \Cref{table:latticeenumeration:algorithmruntime}.
\end{proof}

In the same way, we can derive polynomial space and delay algorithms for related models of behavioral equivalence like simulation~\cite{m-adsp-71} or ready simulation~\cite{bim-bcbt-95,bp-tdine-95}, and some role equivalences such as ecological equivalence~\cite{be-gcpee-92,be-epc-94}, as well as their extensions to more complex network structures.

On the other hand, regular equivalences are naturally expressed as fixed points of a decreasing isotone map in \Cref{equation:isotonemap_bisimulation} restricted to the lattice of equivalences. Following \Cref{table:latticeenumeration:algorithmruntime}, the best worst-case delay achieved by \Cref{alg:nincisotoneenumeration2,alg:backtracksearch} for listing regular equivalences is the exponential bound $\mathcal O(2^n(n^2+m\log n))$, which is obtained by combining \Cref{alg:nincisotoneenumeration2} with the partition refinement algorithm due to Paige and Tarjan~\cite{pt-tpra-87}. Nevertheless, this runtime bound still constitutes a major improvement over a naive brute force approach with a worst-case delay that is super-exponential in~$n$.

\section{Conclusion} \label{section:conclusion}

In this article, we investigated the computational problem of listing the Tarski fixed points of isotone maps on lattices. We first connected the hardness of the problem to structural properties of the given lattice when isotone maps are given as oracles, showing that any deterministic or bounded-error randomized algorithm must perform queries to the given isotone map at least on the order of the width of the underlying lattice to find three or more fixed points. Likewise, we established that the number of queries are asymptotically bounded from below by the number of lower covers of the lattice maximum or upper covers of the lattice minimum for increasing and decreasing isotone maps.

Next, we presented two enumeration algorithms for Tarski fixed points of increasing or decreasing isotone maps. They are based on depth-first (improving on a prior algorithm for a related problem by Echenique~\cite{e-faegs-07}) and flashlight search, two common strategies for enumeration problems. They run in polynomial space on polynomial-height lattices and generally provide substantially better delay bounds than brute-force search. The algorithms build on different assumptions about the underlying lattice: The depth-first algorithm depends on efficient enumeration of lower or upper covers of lattice elements, while flashlight search needs a suitable decomposition of lattice elements into multiple dimensions. Because of these different requirements, it depends on the properties of the specific lattice and the subclass of isotone maps which of these two algorithms are applicable to a particular Tarski fixed point enumeration problem or provides stronger runtime guarantees.

Echenique~\citep{e-faegs-07} previously proposed an algorithm to list Nash equilibria of super-modular games, which can be easily adapted to enumerate Tarski fixed points. Essentially, it traverses the same graph as the depth-first search enumeration algorithm, but does so in breadth-first fashion. However, enumeration in depth-first search is a major improvement, as it ensures that the number of not completely processed but found fixed points is bounded by the height of the fixed point lattice. As a result, depth-first search has stronger delay bounds than Echenique's algorithm.

While there might be other efficient enumeration algorithms for Tarski fixed points, it seems questionable that the established enumeration strategies of proximity~\citep{cu-npdbmseps-19} or reverse search~\citep{af-rse-96} would be suited to this task. These strategies effectively require some way to invert the isotone map, if they are supposed to not just reduce to the depth-first search algorithm or another graph traversal method that moves in the direction of the lattice ordering. But in the worst case, inverting the map for some lattice elements requires testing most or all of the lattice as potential preimages. Therefore, inversion of the map can be just as costly as enumerating Tarski fixed points by brute force, which evaluates the isotone map at each lattice element at most once.

Finally, we applied the presented hardness results and enumeration algorithms to the lattices of equivalences, quasiorders and binary relations. The algorithms run in polynomial space, and we show that they achieve polynomial delay whenever the problem of finding three or more fixed points is neither NP-hard nor has an exponential lower bound on the number of queries. Besides, we witness cases where depth-first search enumeration yields better runtime bounds than flashlight search and vice versa.

Our findings give rise to a polynomial delay and space algorithm for listing bisimulations. Specifically, the suggested algorithm runs with $\mathcal O(n^3m)$ delay and in $\mathcal O(n^2)$ space on a transition system with $n$ states and $m$ transitions. Similarly, we can also obtain polynomial delay and space algorithms for enumerating instances of some related concepts of behavioral equivalence, such as simulations~\citep{m-adsp-71} and ready simulations~\citep{bim-bcbt-95,bp-tdine-95}, but also some models of role equivalence like ecological equivalence~\citep{be-gcpee-92}, for which no polynomial enumeration algorithms appear to have been described in the literature before. 

However, the presented algorithms do not guarantee polynomial delay enumeration for certain other models of behavioral or role equivalence, including the important case of regular equivalence. Yet even in these cases, they typically represent a major improvement in asymptotic delay bounds over brute force.

Since we derived the worst-case query lower bounds and NP-hardness results over all possible isotone maps or specific subclasses thereof, our findings do not imply that enumeration of Tarski fixed points must be as hard for a specific isotone map. Rather, a specific map might allow for algorithms with better performance characteristics, as these algorithms could for instance exploit additional structural features of this map. Notably, this leaves open the complexity of some problems related to regular equivalences, as they are fixed points of one specific decreasing isotone map on the lattice of equivalences. While finding instances of some subclasses of regular equivalences is NP-hard~\citep{fp-cccra-05,rs-hhidi-01}, these and our findings do not necessarily mean that enumerating arbitrary regular equivalences is hard as well.

\section*{Acknowledgments}
The author would like to thank an anonymous reviewer for their very valuable feedback that led to major improvements to this article.

\end{document}